%% file: ArXiv.tex
\documentclass[conference,letterpaper]{IEEEtran}
\usepackage{cite}
\usepackage{graphicx,color,epsfig,rotating}
\usepackage{amsfonts,amsmath,amssymb,bbm,amsthm}
\usepackage{algorithm}
\usepackage[noend]{algpseudocode}
\usepackage{amsmath}
\usepackage{blkarray}
\usepackage{cite}
\usepackage{mdwtab} 
\usepackage{placeins}
\usepackage{psfrag, graphicx}
\usepackage[latin1]{inputenc}
\usepackage{amssymb}
\usepackage{makeidx}
\usepackage{epstopdf}
\usepackage{enumitem}
\usepackage{pstricks}
\usepackage{subcaption}
\usepackage{caption}
\usepackage{bm}
\usepackage{xspace}
\usepackage[nodisplayskipstretch]{setspace} %\setstretch{1.5}
\usepackage{geometry}
\usepackage{cuted}
\newtheorem{construction}{Construction}

\algrenewcommand\algorithmicrequire{\textbf{Input:}}
\algrenewcommand\algorithmicensure{\textbf{Output:}}

\input{macros}

\begin{document}
\title{Multiaccess Coded Caching with Heterogeneous Retrieval Costs}

\author{
    \IEEEauthorblockN{ Wenbo Huang\IEEEauthorrefmark{1}, Minquan Cheng\IEEEauthorrefmark{2}, Kai Wan\IEEEauthorrefmark{1}, Xiaojun~Li\IEEEauthorrefmark{2}, Robert~Caiming~Qiu\IEEEauthorrefmark{1}, and Giuseppe~Caire\IEEEauthorrefmark{3}}
	\IEEEauthorblockA{\IEEEauthorrefmark{1}Huazhong University of Science and Technology, 430074  Wuhan, China,  \{eric\_huang,kai\_wan,caiming\}@hust.edu.cn} 
    \IEEEauthorblockA{\IEEEauthorrefmark{2}Guangxi Normal University,
541004 Guilin, China, chengqinshi@hotmail.com, xiaojun@hotmail.com}
\IEEEauthorblockA{\IEEEauthorrefmark{3} Technische Universit\"at Berlin, 10587 Berlin, Germany,   caire@tu-berlin.de}
}

\maketitle

\begin{abstract}
The multiaccess coded caching (MACC) system, as formulated by Hachem {\it et al.}, consists of a central server with a library of $N$ files, connected to $K$ cache-less users via an error-free shared link, and $K$ cache nodes, each equipped with cache memory of size $M$ files. Each user can access $L$ neighboring cache nodes under a cyclic wrap-around topology. Most existing studies operate under the strong assumption that users can retrieve content from their connected cache nodes at no communication cost. In practice, each user retrieves content from its $L$ different connected cache nodes at varying costs. Additionally, the server also incurs certain costs to transmit the content to the users. In this paper, we focus on  a cost-aware MACC system and aim to minimize the total system cost, which includes cache-access costs and broadcast costs. %Firstly, we propose a novel coded caching framework that employs the scheme of Cheng {\it et al.} as the inner code and adopts superposition coding as the outer code. 
Firstly, we propose a novel coded caching framework based on superposition coding, where the MACC schemes of Cheng \textit{et al.} are layered. Then, a cost-aware optimization problem is derived that optimizes cache placement and minimizes system cost. By identifying a sparsity property of the optimal solution, we propose a structure-aware algorithm with reduced complexity. Simulation results demonstrate that our proposed scheme consistently outperforms the scheme of Cheng {\it et al.} in scenarios with heterogeneous retrieval costs.

\end{abstract}

\begin{IEEEkeywords}
 Coded caching, multiaccess, communication cost, superposition coding, optimization.
\end{IEEEkeywords}

\section{Introduction}
\label{sec: intro}

Coded caching is a promising technique to effectively reduce peak traffic by using local caches and the multicast gains generated by these local caches. Maddah-Ali and Niesen in \cite{Maddah-Ali_centralized} proposed the first well-known coded caching framework, where a central server (owning $N$ files) is connected to $K$ users via error-free shared links. Each user can cache up to $M$ files. The objective is to minimize the normalized amount of transmission for the worst-case over all possible demands. 

Edge caching, especially in the form of multiaccess caching, has recently garnered significant attention due to its ability to improve both spatial and spectral efficiency~\cite{Golrezaei2012Femto,Ji2016Multiaccess}. The authors in \cite{Multi-level} proposed the first multiaccess coded caching (MACC) system with parameters $(K, L, M, N)$ which consists of a server with $N$ files, $K$ cache nodes and $K$ cache-less users, each of which is connected to $L$ neighboring caches in a cyclic wrap-around manner, each cache node can store at most $M \leq \frac{N}{L}$ files.\footnote{\label{foot: 1} If $M > \tfrac{N}{L}$, each user can directly retrieve all requested files from its connected cache nodes, and the problem becomes trivial.} Each user can decode the coded broadcast messages, using all available content in the caches (at no load cost) that the user is connected to. The objective is to minimize the worst-case transmission load. %There are various MACC schemes proposed in \cite{reddy2020rate,serbetci2019multi,MuralidharMAC,MAC} to further improve the load or subpacketization performance. For other multiaccess settings, there are some studies on the scheme under the coded placement in \cite{coded_placement1,coded_placement2}. Many studies focus on the access topologies such as across revisable design access \cite{Cross_resolvable}, combinatorial design access \cite{Multiaccess_topologies} and so on. 
Subsequent works have proposed various MACC schemes to further improve the transmission load or reduce the subpacketization level, e.g.,~\cite{reddy2020rate,serbetci2019multi,MuralidharMAC,MAC}. For other multiaccess caching settings, coded placement strategies have been investigated in~\cite{coded_placement1,coded_placement2}. Moreover, different cache-access topologies have been studied, including cross-resolvable designs~\cite{Cross_resolvable} and combinatorial designs~\cite{Multiaccess_topologies}, among others.

As discussed above, most existing studies rely on an idealized assumption that the cost incurred when a user retrieves content from its connected cache nodes is negligible. However, this assumption rarely holds in practical wireless or distributed systems~\cite{yao2019mobile,li2018survey}. In edge or fog computing scenarios~\cite{sengupta2017fog,azimi2018online}, user-to-cache communication may involve wireless hops, backhaul links, or non-negligible access delays, all of which can contribute significantly to the overall system cost.

In this paper, we consider the MACC with heterogeneous retrieval costs. We study a $(K,\boldsymbol{\mu} = (\mu_1, \mu_2, \ldots, \mu_L),\rho, M, N)$ MACC system, where each user retrieves one file from its $l$-th cyclic right closed cache node at a cost $\mu_l$ (referred to as cache-access cost) for each integer $1\leq l\leq L$, and the server broadcasts one file at a cost $\rho$ (referred to as broadcast cost). Our objective is to minimize the total communication cost under worst-case demands, defined as the sum of the total server's broadcast cost and the total users' cache-access costs. The system model is presented in Fig.~\ref{multiaccess-system}. %By explicitly modeling the heterogeneous access costs $\boldsymbol{\mu}$ and the broadcast cost $\rho$, the MACC abstraction more faithfully reflects real distributed systems with asymmetric links and cost budgets. 
When $\boldsymbol{\mu} =\mathbf{0}$ and $\rho=1$, the MACC system reduces to the classical MACC system in prior works~\cite{MAC,Multi-level,reddy2020rate,serbetci2019multi,MuralidharMAC}, and our objective reduces to the transmission load, which accounts only for the  broadcast cost. Thus, our MACC system can be viewed as a generalization of the original MACC system.

\iffalse
\subsection{Main Contributions}
\label{sec: main con}

We establish a heterogeneous-retrieval-cost framework in MACC system. Using a superposition-based coding architecture together with a tailored structural optimization, we obtain the following main results.

\begin{itemize}
\item 
Based on MACC scheme in~\cite{MAC}, we propose a superposition-based coding that layers coded caching schemes according to the different access levels (i.e., the number of cache nodes accessible
to users). The proposed scheme maintains the advantage of MACC scheme in~\cite{MAC} to achieve the maximum local caching gain\footnote{ the cached contents stored at any $L$ neighboring cache nodes are different such that each user can totally retrieve $LM$ files from the connected cache nodes}, and provides a general scheme by tuning the superposition ratio under arbitrary heterogeneous costs. % {\color{red}By tuning the superposition ratio, this construction convexifies the cost-performance tradeoff, enabling near-optimal placement-delivery under heterogeneous costs and yielding strictly lower total communication cost whenever the optimizer lies between discrete access levels.} 
\item To minimize the communication cost, we formulate an optimization problem to determine the superposition ratio of each access level. We prove that an optimal solution to the proposed optimization problem is sparse: the number of nonzero superposition ratio is at most two. Thus, we propose a structure-aware algorithm to solve it. 
\end{itemize}
\fi

%\subsection{Main Contributions}
%\label{sec:main_con}

We establish a heterogeneous retrieval-cost framework for the MACC system. The main contributions of this paper are summarized as follows. %By employing a superposition-based coding architecture together with a tailored structural optimization, the main contributions of this paper are summarized as follows.
\begin{itemize}
    \item  By applying a superposition strategy into the MACC scheme in~\cite{MAC}, we develop a new scheme that layers multiple coded caching schemes according to different access levels. %i.e., the numbers of cache nodes accessible to each user. 
     The proposed scheme preserves the key advantage of the MACC scheme in~\cite{MAC} in achieving the maximum local caching gain,\footnote{The cached contents stored at any $L$ neighboring cache nodes are mutually distinct, such that each user can retrieve up to $LM$ files from its connected cache nodes.} and further provides a flexible framework that adapts to arbitrary heterogeneous cost settings by  tuning the superposition weights.
    \item To minimize the total communication cost, we formulate an optimization problem to determine the superposition weights corresponding to different access levels. We prove that an optimal solution to the proposed problem exhibits a sparsity property, namely that at most two superposition weights are nonzero. Exploiting this structural insight, we develop a structure-aware algorithm with reduced computational complexity.
\end{itemize}

%\subsection{Notation Convention}

%\begin{itemize}
In this paper, we will use the following notations. Bold capital letter, bold lower case letter, and curlicue font will be used to denote array, vector, and set, respectively. We assume that all the sets are ordered increasingly, and $|\cdot|$ is used to represent the cardinality of a set or
the length of a vector. For any positive integers $a$, $b$, $t$, $q$ with $a<b$ and $t\leq b $, and any nonnegative set $\mathcal{V}$, let $[a:b]=\{a,a+1,\ldots,b\}$, especially $[1:b]$ be shorten by $[b]$; if $a$ is not divisible by $q$, $\langle a\rangle_{q}$ denotes the least non-negative residue of $a$ modulo $q$. Otherwise,  $\langle a\rangle_{q}:=q$. $\mathcal{V}[h]$ represents the $h^{th}$ smallest element of $\mathcal{V}$, where $h\in[|\mathcal{V}|]$.

\section{Problem Statement}
\label{sec: problem}

\subsection{Multiaccess Caching Model with Heterogeneous Retrieval Costs}
\label{sec: system}

In a $(K,\boldsymbol{\mu} = (\mu_1, \mu_2, \ldots, \mu_L),\rho, M, N)$ MACC (MACC) system, users want to retrieve files from the server and cache nodes. One server contains $N$ equal-length files    in the library $\mathcal{W} = \{W_1, W_2, \ldots, W_N\} $, and is connected to $K$ cache-less users $\{U_1, U_2, \ldots, U_K\}$ through a broadcast shared-link. There are also 
  $K$ cache nodes $\{C_1, C_2, \ldots, C_K\}$, each of which contains $M$ files, subject to $M \leq \tfrac{N}{L}$.  User $U_k$ where $k \in [K]$  is connected to $L$ neighboring cache nodes in a cyclic wrap-around manner, i.e., $\{C_k, C_{\langle k+1\rangle_K}, \ldots, C_{\langle k+L-1\rangle_K}\}$. 
  Each user can retrieve the cached content from its connected cache-node with some cost.
  Retrieving one file from the $l^{\text{th}}$ connected cache node $C_{\langle k+l-1\rangle_K}$ incurs a unit cost of $\mu_l$, for $l\in[L]$, while receiving one file directly from the server incurs a unit cost of $\rho$. Thus, the vector $\boldsymbol{\mu}=(\mu_1,\mu_2,\ldots,\mu_L)$ characterizes the cache-access costs, and $\rho$ denotes the server broadcast cost. The system model is presented in Fig.~\ref{multiaccess-system}.
\begin{figure}[htbp]
    \centering
        \includegraphics[width=\linewidth]{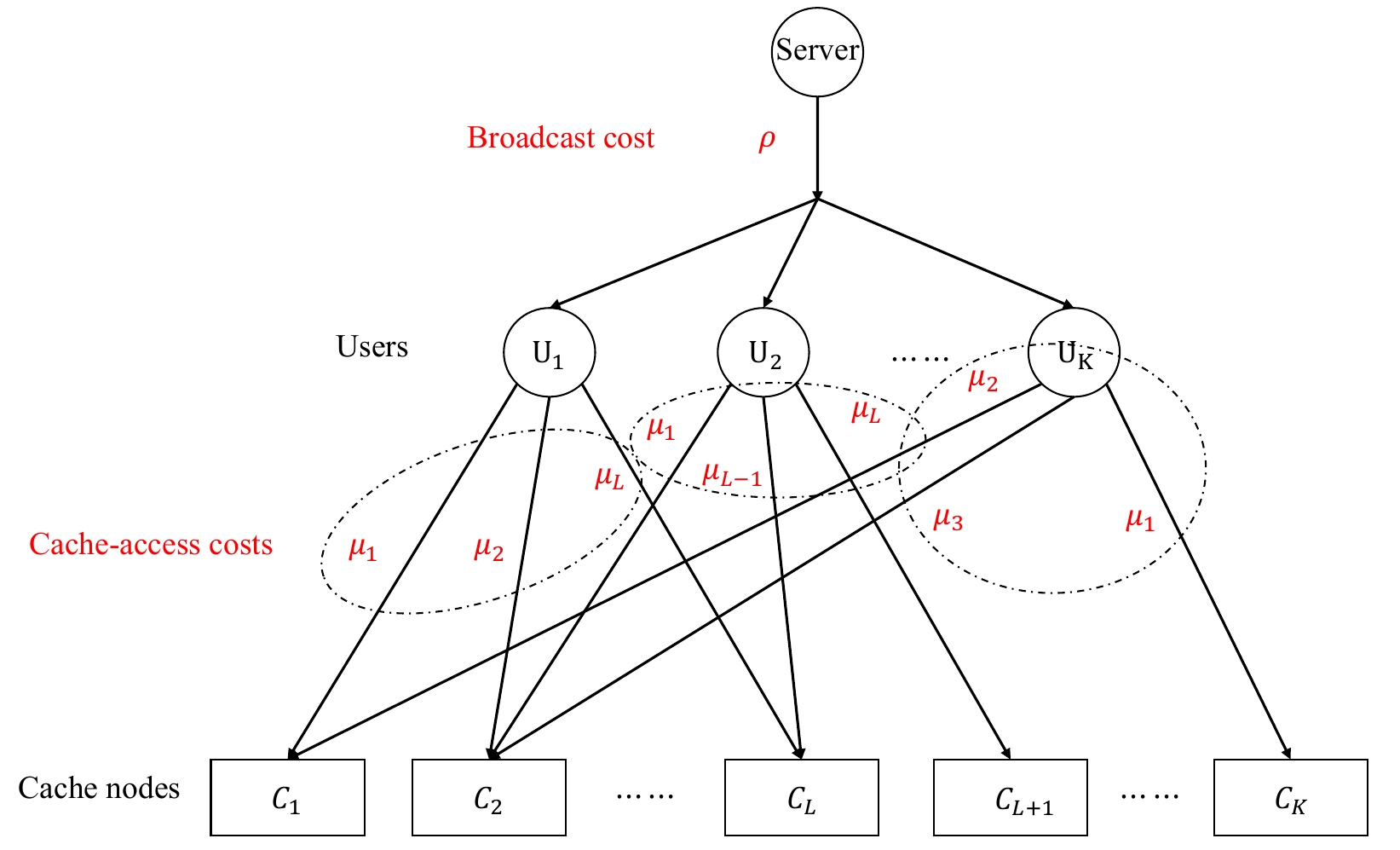}
    \caption{The $(K,L,M,N,\rho,v)$    multiaccess  coded caching system.}   \label{multiaccess-system}
\end{figure}

The $(K,\boldsymbol{\mu},\rho,M,N)$ MACC model consists of two phases:

 {\bf Placement phase:} Each file is divided into $F$ packets of equal size, and each cache node caches $MF$ packets of files. Each user $U_k$ can retrieve the content stored at its connected cache nodes.

 {\bf Delivery phase:} Each user $U_k$ randomly requests one file $W_{d_k}$. Based on the request vector $\mathbf{d} = \left(d_1, d_2, \ldots, d_K\right)$ and the retrieved content of each user, the server broadcasts a message containing $S_{\mathbf{d}}$ packets to all users satisfying each user's request.

The objective is to minimize the worst-case communication cost incurred in the system, which includes both the cache-access cost and the server broadcast cost.
\begin{align}
\label{eq: cost-def}
R = \max\limits_{\mathbf{d} \in [N]^{K}} \frac{\sum_{l=1}^{L}\mu_l\Delta_{\mathbf{d},l} + \rho S_{\mathbf{d}}}{F},
\end{align}
where $\Delta_{\mathbf{d},l}$ represents the total number of the retrieval packets from the $l^{\text{th}}$ cache node.

Most existing works on MACC~\cite{MAC, MuralidharMAC, brunero2022fundamental} neglect the cache-access cost, i.e., by setting $\boldsymbol{\mu}=\mathbf{0}$ and $\rho=1$. In this case, the system degenerates into a $(K,L,M,N)$ MACC model in \cite{Multi-level}, and the communication cost reduces to the worst-case load
\begin{align}
    R = \max\limits_{\mathbf{d} \in [N]^{K}} \frac{ S_{\mathbf{d}}}{F}.
\end{align}

\subsection{MACC Scheme}
\label{sec: MACC}

To the best of our knowledge, all existing MACC schemes under uncoded placement and one-shot delivery strategies can be represented by the following three types of arrays. 
\begin{itemize}
    \item An $F \times K$ node-placement array $\mathbf{C}$ consists of star and null, where $F$ and $K$ represent the subpacketization of each file and the number of cache nodes,  respectively. The entry located at position $(j,k)$ is star if and only if the $k^{\text{th}}$ cache node  caches  the $j^{\text{th}}$ packet of each $W_n$ where $n\in [N]$.
    \item An $F \times K$ user-retrieve array $\mathbf{U}$ consists of star and null, where $F$ and $K$ represent the subpacketization of each file and the number of users, respectively. The entry at position $(j,k)$ is star if and only if the $k^{\text{th}}$ user can retrieve the $j^{\text{th}}$ packet  of each $W_n$ where $n\in [N]$.
    \item An $F \times K$ user-delivery array $\mathbf{Q}$ consists of $\{*\}\cup[S]$, where $F$, $K$,  and the stars in $\mathbf{Q}$ have the same meaning as that of $\mathbf{U}$. Each integer represents a broadcast message, and $S$ represents the total number of broadcast messages transmitted in the delivery phase. Specifically, for any demand vector $\mathbf{d}$, the server transmits the following coded message
\begin{align}
\label{eq-delivery-strategy}  X_s = \oplus_{\mathbf{Q}(j,k)=s} W_{d_k,j},
\end{align} for each time slot $s\in[S]$.  
\end{itemize}

The authors in \cite{yan2017placement} showed that the transmission is a one-shot strategy if the user-delivery array $\mathbf{Q}$ satisfies the following conditions. 
\begin{lemma}[\cite{yan2017placement}]
\label{lem-sufficient-Q}   
The transmission generated by \eqref{eq-delivery-strategy} is of the one-shot type if, for every two different entries $\mathbf{Q}(j_1,k_1)$ and $\mathbf{Q}(j_2,k_2)$ in $\mathbf{Q}$, when $\mathbf{Q}(j_1,k_1) = \mathbf{Q}(j_2,k_2) = s$ is an integer, the following conditions hold.
\begin{enumerate}
\item [C1.] $j_1 \neq j_2, k_1 \neq k_2$, i.e., they lie in distinct rows and distinct columns;
\item [C2.] $\mathbf{Q}(j_1,k_2) = \mathbf{Q}(j_2,k_1) = *$, i.e., the corresponding $2 \times 2$ subarray formed by rows $j_1, j_2$ and columns $k_1,k_2$ is either 
        $$ \left(\begin{array}{cc}
         s & *  \\
         * & s 
 \end{array}\right) ~\text{or}~ \left(\begin{array}{cc}
         * & s  \\
         s & * 
 \end{array}\right). $$
 \end{enumerate}	\hfill $\square$ 
\end{lemma}
Assume that in $\mathbf{Q}$ there are $l$ occurrences of the same integer $s$, located in $\{(j_u,k_u)\}^{l}_{u=1}$, where $1 \le j_u \le F$ and $1 \le k_u \le K$. By C1, all $j_u$ and $k_u$ are distinct, so the subarray formed by rows $\{j_1,\dots,j_l\}$ and columns $\{k_1,\dots,k_l\}$ is $l \times l$. By repeatedly applying C2 to any pair $(u,v)$ with $u\neq v$, we have $\mathbf{Q}(j_u,k_v)=*$. Hence, up to simultaneous row/column permutations, this subarray has the form
\begin{eqnarray}\label{Eqn_Matrix_1}
    \left(\begin{array}{ccc}
      s &  \cdots & *\\
            \vdots  &\ddots & \vdots\\
      * &  \cdots & s
    \end{array}\right).
\end{eqnarray}

According to the delivery rule \eqref{eq-delivery-strategy}, at time slot $s$ (with $1\le s \le S$) the server transmits $\oplus_{u\in[l]} W_{d_{k_u},\,j_u}$. 
From \eqref{Eqn_Matrix_1}, in column $k_v$ all entries within the selected $l$ rows are $\ast$ except at row $j_v$. By the placement associated with $\ast$, user $k_v$ can retrieve all packets $W_{d_{k_u},\,j_u}$ for $u\neq v$. Therefore, it subtracts those from the broadcast and recovers its desired packet $W_{d_{k_v},\,j_v}$. Consequently, all users decode successfully for any demand vector.

The authors in~\cite{MAC} proposed the following construction, which has the maximum local caching gain and a coded caching gain as the MN scheme.
%obtains the optimal worst-case load under the maximum local caching load.
\begin{construction}[MACC arrays, \cite{MAC}]
\label{constr-3-arrays}For any positive integers $K'$, $t$ and $L$, we can construct the node-placement array, user-retrieve array, and user-delivery array, denoted by 
\begin{align*}
\mathbf{C}&=(\mathbf{C}((\mathcal{T},g),k))_{\mathcal{T}\in{[K']\choose t},g,k\in[K]},\\
\mathbf{U}&=(\mathbf{U}((\mathcal{T},g),k))_{\mathcal{T}\in{[K']\choose t},g,k\in[K]},\\
\mathbf{Q}&=(\mathbf{Q}((\mathcal{T},g),k))_{\mathcal{T}\in{[K']\choose t},g,k\in[K]},
\end{align*} which can be used to generate a $(K,L,M,N)$
 MACC scheme where $K = K' +t(L-1)$, with the memory ratio $\frac{M}{N}=\frac{t}{K}$, subpacketization $F=\binom{K'}{t}K$ and transmission load $R=\frac{K-tL}{t+1}$ as follows.  For any $\mathcal{T}\in{[K']\choose t}$, $g,k\in[K]$, the entries  
\begin{align}
\mathbf{C}((\mathcal{T},g),k)& = 
\begin{cases}
   * & \text{if}\ k \in \mathcal{C}_{\mathcal{T},g}\\
 null & \text{otherwise},
    \end{cases}\label{eq-cach-node-C}\\
  \mathbf{U}((\mathcal{T},g),k)& = 
    \begin{cases}
        * & \text{if}\  k \in \mathcal{U}_{\mathcal{T},g} \\
        null & \text{otherwise},
    \end{cases}  \label{eq-cach-user-U}  \\
\mathbf{Q}((\mathcal{T},g),k)&= 
    \begin{cases}
        * & \text{if}\  k \in \mathcal{U}_{\mathcal{T},g} \\
        (\mathcal{T} \cup \{\psi_{\mathcal{T},g}(k)\},g) & \text{otherwise},
    \end{cases} \label{eq-delivery-user-Q} 
\end{align}
where
\begin{align} 
\mathcal{C}_{\mathcal{T},g}&= \{ \langle \mathcal{T}[h] + h(L-1) +(g-1) \rangle_{K}  : h \in [t]   \}\label{eq-row-cache}  \\
\mathcal{U}_{\mathcal{T},g}&= \left\{ \langle \mathcal{T}[h] + (h-1)(L-1) + (g-1)\rangle_{K},  \right. \nonumber\\  & \ \ \ \ \ \langle \mathcal{T}[h] + (h-1)(L-1) + (g-1) + 1\rangle_{K},  \nonumber\\
&\ \ \ \  \left.  \ldots, \langle\mathcal{T}[h] + h(L-1) +(g-1)\rangle_{K}: h \in [t]\right\}, \label{eq-row-user}
\end{align}
and $\psi_{\mathcal{T},g}(\cdot)$ is a one-to-one mapping function which maps the integer $k$ into the ${\langle n + (g-1) \rangle_{K'}}$-th entry of $[K'] \setminus \mathcal{T}$, where $k$ is the $n^{\text{th}}$ entry of $[K]\setminus \mathcal{U}_{\mathcal{T},g}$.
\hfill $\square$
\end{construction}

To illustrate the Construction~\ref{constr-3-arrays}, we provide an example in the Appendix~\ref{sec: Cons1}. By Construction~\ref{constr-3-arrays}, the authors in \cite{MAC} obtained the following result.
\begin{lemma}[\cite{MAC}]
    \label{thm: MACC}
    For any positive integers $(K',t,L)$, there exists a $(K,L,M,N)$ MACC scheme where $K=K'+t(L-1)$, $\frac{M}{N} = \frac{t}{K}$ with the worst-case load $R = \frac{K-tL}{t+1}$.
    
    \hfill $\square$
\end{lemma}

\section{Main Results}
\label{sec: Main}

%We propose a superposition MACC scheme and characterize the optimal communication cost through an optimization problem. To solve the optimization problem, we propose a structure-aware algorithm. 

Firstly, we propose a new coded caching framework based on superposition coding, where the MACC schemes of Cheng \textit{et al.}~\cite{MAC} are layered across different access levels. Then, a cost-aware optimization problem is derived that optimizes cache placement and minimizes system cost. To solve the problem, we develop a structure-aware algorithm by exploiting the sparsity property of the optimal solution.

\subsection{Superposition MACC scheme}
\label{sec: super}

\begin{theorem}\label{thm: com}
For any positive real vectors $\boldsymbol{\alpha} := (\alpha_1,\alpha_2, \ldots, \alpha_L)$ and $\boldsymbol{\gamma} := (\gamma_1,\gamma_2, \ldots, \gamma_L)$, the optimal communication cost of the $(K,\boldsymbol{\mu},\rho,M,N)$ superposition MACC scheme based on~\cite{MAC}  is characterized by the solution of the following optimization problem (where $M \leq \tfrac{N}{|\boldsymbol{\mu}|} = \tfrac{N}{L}$):
\begin{align}
\underset{\boldsymbol{\alpha,\gamma}}{\text{minimize }} \quad & R =   \sum_{l\in [L]}\alpha_l \left[ K\frac{1-l\gamma_l}{K\gamma_l+1}\rho + K\gamma_l \sum_{\ell=1}^{l} \mu_{\ell} \right. \nonumber \\
& \quad \quad \left. + K^2\gamma_l\frac{1-l\gamma_l}{K\gamma_l+1}(\mu_1+\mu_{l}) \right] \nonumber \\[1ex]
\text{subject to} \quad 
& 0 \leq \alpha_l \leq 1, \quad \forall l \in [L]; \label{eq: bound1} \\ 
& \tfrac{1}{K} \leq \gamma_l \leq \tfrac{1}{K}\Big\lfloor\tfrac{K}{l}\Big\rfloor, \quad \forall l \in [L]; \label{eq: bound2} 
\end{align}
\begin{align}
& \sum_{l\in [L]}\alpha_l = 1; \label{eq: con1} \\
& \sum_{l\in [L]}\alpha_l\gamma_l = \tfrac{M}{N}. \label{eq: con2}
\end{align} \label{eq: opt}  \hfill $\square$
\end{theorem} 

%Now, we present the superposition coding scheme, as well as the established process of the optimization problem. When $\boldsymbol{\mu} \neq \mathbf{0}$ and $\rho \neq 1$, employing a caching scheme with access to $l \in [L]$ cache nodes results in different cache-access costs. Intuitively, increasing the access level (i.e., the number of cache nodes accessible to users) generally provides a higher caching gain and requires fewer broadcast messages, since users can exploit more side information from the caches. However, this also leads to higher cache-access costs. To minimize the communication cost, we can utilize the superposition coding method to combine heterogeneous caching schemes with varying access levels. 

We now present the proposed superposition MACC scheme and outline the derivation of the associated optimization problem. When costs are heterogeneous, i.e., $\boldsymbol{\mu} \neq \mathbf{0}$, and $\rho \neq 1$, MACC schemes operating at different access levels incur distinct tradeoffs between broadcast cost and cache-access cost. Specifically, a higher access level enables users to exploit more cached side information and thereby reduces the required broadcast transmissions, while simultaneously increasing the cache-access cost due to accessing more cache nodes. This observation motivates the use of superposition coding to combine multiple MACC schemes with different access levels in order to minimize the total communication cost.

Specifically, we divide each file into $L$ subfiles, denoted by $W_{n}^{l}$ for $n \in [N]$ and $l \in [L]$. Each subfile has a size ratio $\alpha_{l}\in [0,1]$ equal to the ratio of its length and the length of the file, with the constraints in~\eqref{eq: bound1} and~\eqref{eq: con1}: $0 \leq \alpha_l \leq 1$ and $\sum_{l \in [L]} \alpha_{l} = 1$. For each access level $l \in [L]$, we employ the coding scheme in~\cite{MAC}. Based on integers $(K_{l},t_{l},l)$, we construct a $(K,l,M_{l},N)$ MACC scheme for the subfile set $\mathcal{W}^{l} = \{W_{n}^{l} \mid n \in [N]\}$, where
$K = K_{l} + t_{l}(l-1) $ and $ \frac{K M_{l}}{N} = t_{l}$.
The corresponding caching ratio is denoted by $\gamma_{l} = \tfrac{M_{l}}{N}$, with the constraints~\eqref{eq: bound2}~\eqref{eq: con2}: $\tfrac{1}{K} \leq \gamma_l \leq \tfrac{1}{K}\Big\lfloor\tfrac{K}{l}\Big\rfloor$ and $\sum_{l\in[L]}\alpha_l\gamma_l=\frac{M}{N}$. 
Note that $\alpha_{l}$ also represents the weight of the superposition coding for the $l$-th coding scheme on the $l$-th subfiles of all files, i.e., the $(K,l,M_{l},N)$ MACC scheme.

Next, let us consider the communication cost in the $l$-level scheme and the detailed computations are included in Appendix~\ref{sec: optimization problem establish}. By Lemma~\ref{thm: MACC}, the server broadcasts $\frac{K_l- l t_l}{t_l+1} = K \frac{1 - l\gamma_{l}}{K\gamma_{l} + 1}$ messages with size ratio $\alpha_l$ under the cost $\rho$. So the broadcast cost is 
\begin{align} 
R^l_{b} = \alpha_l \cdot K \frac{1 - l\gamma_{l}}{K\gamma_{l} + 1} \cdot \rho. \label{eq: Rb}
\end{align} The total cache-access cost can be divided into two parts.

{\bf Retrieving requested files}: In the scheme of~\cite{MAC}, the local caching gain is maximized, and thus each user can retrieve exactly $K\gamma_l$ files from its connected cache nodes. The corresponding cost is
    \begin{align} 
        R^l_{c_1} = \alpha_l \cdot K\gamma_l \cdot \sum_{\ell=1}^{l} \mu_{\ell}. \label{eq: Rc1}
    \end{align}

{\bf Retrieving decoding files}: There is an important character that should be mentioned specifically. Each user can only use the $1^{\text{st}}$ closed cache node and the $L^{\text{th}}$ closed cache node to access packets and decode the required packets from the coded packets sent by the server. Each user retrieves $t_l$ decoding files from the first and $l^{\text{th}}$ accessible cache nodes respectively, to decode the requested files from $K \frac{1 - l\gamma_{l}}{K\gamma_{l} + 1}$ broadcast messages. This fact is shown in the second investigation of Example \ref{ex: MACC} in Appendix~\ref{sec: Cons1}. The general proof is included in Appendix \ref{sec: costs}. So the retrieving access cost is 
        \begin{align}  
        R^l_{c_2} = \alpha_l \cdot K\gamma_l \cdot K \frac{1-l\gamma_l}{K\gamma_l+1} \cdot (\mu_1+\mu_{l}). \label{eq: Rc2}
    \end{align}

From~\eqref{eq: Rb}-\eqref{eq: Rc2}, the total cost in Theorem~\ref{thm: com} is obtained. 

\subsection{Structure-aware Algorithm}
\label{sec: algorithm}

Solving the optimization problem described in Theorem~\ref{eq: opt} is challenging mainly due to the following two reasons:

(i) the feasibility region is non-convex due to the bilinear constraint~\eqref{eq: con2}, and 

(ii) the objective function is non-convex.

To address this issue, we propose an enhanced sequential quadratic programming (SQP)~\cite{sqp} algorithm that leverages the properties of the optimal solution. 

\paragraph{Property of the Optimal Solution}

In this paragraph, we will provide Proposition 1, which leads to the main improvement of our proposed algorithm. 

\iffalse
\begin{proposition}[Two-Point Structure of the Optimal Solution] 
For any parameters ($K$, $\boldsymbol{\mu}$, $\rho$, $M$, $N$) in Theorem \ref{thm: com}, there exists an optimal solution $(\boldsymbol{\alpha},\boldsymbol{\gamma})$ with support on at most two index-wise pairs $(\alpha_i,\gamma_i)$ and $(\alpha_j,\gamma_j)$, for some $i,j \in [L]$.  
\hfill $\square$
\end{proposition}
\fi

\begin{proposition}[Two-Point Structure of the Optimal Solution]
For any parameters $(K, \boldsymbol{\mu}, \rho, M, N)$ in Theorem~\ref{thm: com}, there exists an optimal solution $(\boldsymbol{\alpha}, \boldsymbol{\gamma})$ and a subset $\mathcal{I} \subseteq [L]$ with $|\mathcal{I}| \leq 2$ such that $\alpha_l = 0$ for all $l \in [L] \setminus \mathcal{I}$. Equivalently, the optimal solution has support on at most two indices.
\hfill $\square$
\end{proposition}

\begin{proof}
To prove Proposition 1, let us first consider the optimization with respect to the superposition weight $\boldsymbol{\alpha}$. For a fixed caching ratio $\boldsymbol{\gamma}$, the original problem is degenerated into a standard form linear programming problem: 
\begin{align} 
\underset{\boldsymbol{\alpha}}{\text{minimize }} \quad & R(\boldsymbol{\alpha}\mid\boldsymbol{\gamma})  \label{eq: AO1}   \\
\text{subject to } \quad 
&0\leq\alpha_l\leq 1, \forall l \in [L]; \quad  \sum_{l=1}^{L}\alpha_l = 1;  \nonumber \\
& \sum_{l=1}^{L}\alpha_l\gamma_l=\frac{M}{N}. \nonumber  
\end{align}  

The Fundamental Theorem of Linear Programming~\cite{linear_optimization} shows  that for a linear programming problem $\min \{ \mathbf{c}^\top \mathbf{x} \mid \mathbf{A}\mathbf{x} = \mathbf{b}, \mathbf{x} \ge \mathbf{0} \}$, we have the following statements:
\begin{itemize}
    \item If there exists a feasible solution, there exists a basic feasible solution;
    \item if there exists an optimal solution, there exists an optimal basic feasible solution;
    \item the number of non-zero elements in a basic feasible solution should be no more than $m$, where $m$ is the number of linear equality constraints.
\end{itemize}

The optimization problem~\eqref{eq: AO1} is a linear programming problem with $2$ linear equality constraints. Thus, there is an optimal solution of~\eqref{eq: AO1} with at most $2$ non-zero elements in vector $\boldsymbol{\alpha}$. 

Suppose the optimal solution of the optimization problem in Theorem~\ref{thm: com} is $\left( \boldsymbol{\alpha}^{*}, \boldsymbol{\gamma}^{*} \right)$. Since $\boldsymbol{\alpha}^{*}$ is the optimal solution, it must be the optimal solution of the linear programming problem~\eqref{eq: AO1} under the caching ratio $\boldsymbol{\gamma}^{*}$. Then, we conclude that there exists an optimal solution $(\boldsymbol{\alpha},\boldsymbol{\gamma})$ with support on at most two index-wise pairs $(\alpha_i,\gamma_i)$ and $(\alpha_j,\gamma_j)$, for some $i,j \in [L]$.  
\end{proof}

\paragraph{Optimization Algorithm}

Based on the observation in Proposition 1, after choosing a specific pair of non-zero $(\alpha_i, \alpha_j)$, where $i,j \in [L]$, we can rewrite the optimization problem into a simple formulation with two pairs of $(\alpha_i,\gamma_i)$ and $(\alpha_j, \gamma_j)$. We further simplify the objective into the function with two variables $\gamma_i,\gamma_j$ like~\eqref{eq: simple form}, 
\begin{align} \label{eq: simple form}
\underset{\gamma_i,\gamma_j}{\text{minimize }} \quad & R = \frac{(\frac{M}{N} - \gamma_i)\left( 
A_i \gamma_i^2 + B_i \gamma_i  
+  K\rho \right)}{(\gamma_j - \gamma_i)(K\gamma_i+1)} \nonumber\    \\  & \quad + \frac{(\gamma_j - \frac{M}{N})\left(A_j \gamma_j^2 + B_j \gamma_j  
+  K\rho \right)}{(\gamma_j - \gamma_i)(K\gamma_j+1)}  \\
\text{subject to } \quad
& \frac{1}{K}\leq\gamma_i\leq \frac{1}{K}\lfloor\frac{K}{i}\rfloor; \frac{1}{K}\leq\gamma_j\leq \frac{1}{K}\lfloor\frac{K}{j}\rfloor, \nonumber  
\end{align}
where $A_i = K^2 \sum_{\ell=1}^{i} \mu_{\ell} - K^2 i(\mu_1 + \mu_i)$ and $B_i = K \sum_{\ell=1}^{i} \mu_{\ell} - Ki\rho + K^2(\mu_1 + \mu_i)$. Although the feasible region is convex now, the objective function is still non-convex. %However, it is continuous and smooth, possessing continuous first- and second-order derivatives within the feasible domain. 
Thus, we employ the SQP algorithm. Since the simplified problem involves only two decision variables $\gamma_i,\gamma_j$ with box constraints, we solve it via an SQP method with a standard globalization strategy (merit-function line search/filter). By restricting the iterates away from the singular set $\gamma_i = \gamma_j$, the objective is continuously differentiable on the feasible set. In addition, the limit point satisfies the second-order sufficient conditions; thus, the SQP method enjoys quadratic convergence with exact Hessian updates. Because the SQP subproblem is two-dimensional, the algorithm has constant per-iteration complexity.

To avoid the prohibitive complexity of exhaustive search over all index pairs, we propose a greedy index search method with restricted neighborhood exploration, whose 
  main intuition is as follows. 
Starting from an initial feasible index pair, the algorithm iteratively improves the solution by locally substituting one index at a time. In each iteration, only a small candidate set of size $B = \frac{L}{2}$ is examined, and a first-improvement strategy is adopted to terminate the neighborhood search as soon as a better solution is found. Moreover, the underlying SQP subproblems are warm-started using the incumbent solution.

The overall algorithm, referred to as Greedy Index Search with Restricted Neighborhood and Warm-Started SQP, is provided in Algorithm~\ref{alg: greedy_search_warm}. The Algorithm complexity is $\mathcal{O}(IBT)$, where $I$ denotes the number of outer iterations in which the incumbent index pair is strictly improved, and $T$ is the number of iterations in the SQP process. 

\begin{algorithm}
\caption{Greedy Index Search with Restricted Neighborhood and Warm-Started SQP}
\label{alg: greedy_search_warm}
\begin{algorithmic}[1]
\Require System parameters $L, K, M, N, \rho, \boldsymbol{\mu}$, neighborhood budget $B = \frac{L}{2}$
\Ensure $(i^\star, j^\star)$, $(\alpha_{i^\star}, \alpha_{j^\star})$, $(\gamma_{i^\star}, \gamma_{j^\star})$, and $R^\star$

\State \textbf{Initialization:}
\State Select an initial feasible index pair $(i^{(0)}, j^{(0)})$ 
\State Solve the SQP subproblem for $(i^{(0)}, j^{(0)})$ to obtain $(\gamma_{i^{(0)}}, \gamma_{j^{(0)}})$ and objective value $R^{(0)}$
\State Set $(i^\star, j^\star)\gets (i^{(0)}, j^{(0)})$, $R^\star \gets R^{(0)}$, and $(\gamma_{i^\star}, \gamma_{j^\star}) \gets (\gamma_{i^{(0)}}, \gamma_{j^{(0)}})$

\Repeat
    \State \textbf{Restricted candidate set construction:}
    \State Build $\mathcal{K}\subseteq \{1,\dots,L\}\setminus\{i^\star,j^\star\}$ with $|\mathcal{K}|=B$
    \State Set \texttt{improved} $\gets$ \textbf{false}

    %\Statex \Comment{First-improvement scan over restricted neighborhood: }
    \For{each $k\in\mathcal{K}$}
        \For{each candidate pair $(p,q)\in\{(k,j^\star),(i^\star,k)\}$}

            \State \textbf{Warm-start: }
            \State Set initial point $(\gamma_p^{(0)},\gamma_q^{(0)}) \gets (\gamma_{i^\star},\gamma_{j^\star})$
            \State Solve the SQP subproblem for $(p,q)$ initialized at $(\gamma_p^{(0)},\gamma_q^{(0)})$ to obtain $(\gamma_p,\gamma_q)$ and $R_{p,q}$

            \If{$R_{p,q} < R^\star$}
                \State $(i^\star, j^\star) \gets (p,q)$, $R^\star \gets R_{p,q}$
                \State $(\gamma_{i^\star}, \gamma_{j^\star}) \gets (\gamma_p,\gamma_q)$
                \State \texttt{improved} $\gets$ \textbf{true}
                \State \textbf{break} %\Comment{First-improvement: stop scanning further candidates}
            \EndIf
        \EndFor
        \If{\texttt{improved}}
            \State \textbf{break}
        \EndIf
    \EndFor
\Until{\texttt{improved}=\textbf{false}}

\Statex
\State \textbf{Post-processing:}
\State Compute $(\alpha_{i^\star}, \alpha_{j^\star})$ from $(\gamma_{i^\star}, \gamma_{j^\star})$
\State \Return $(i^\star, j^\star)$, $(\alpha_{i^\star}, \alpha_{j^\star})$, $(\gamma_{i^\star}, \gamma_{j^\star})$, $R^\star$
\end{algorithmic}
\end{algorithm}

\subsection{Performance Analysis}
\label{sec: performance}

\begin{figure}[htbp]
    \vspace{-4pt}
    \centering
    \includegraphics[width=0.7\linewidth]{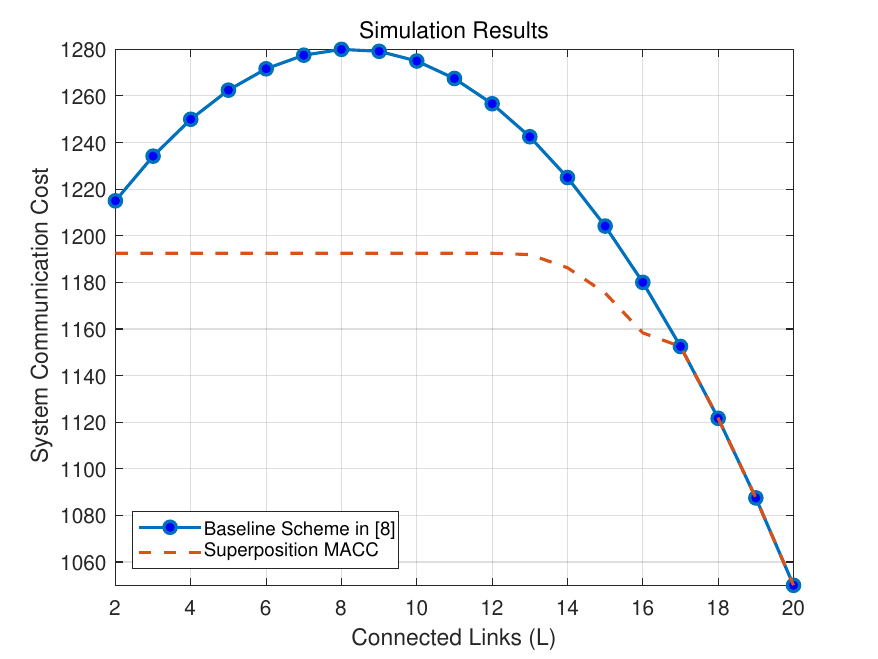}
    \caption{\small Performance comparison between the baseline Scheme and the proposed superposition scheme.}  
    \vspace{-6pt}
    \label{fig:performance_comparison}
\end{figure}

In the simulations, we consider the $(K=100, N=100, M=5, \boldsymbol{\mu} = (1,2,\ldots,L), \rho = 65)$ MACC systems, where each user is connected to $L$ cache nodes with  $L \in [2:20]$.   A baseline scheme is deploying the $(K,L,M,N)$ MACC scheme in~\cite{MAC} directly into the considered $(K,\boldsymbol{\mu},M,N,\rho)$ MACC model, then the communication cost can be written as 
$
    R =  K\frac{1-L\frac{M}{N}}{K\frac{M}{N}+1}\rho + K\frac{M}{N} \sum_{\ell=1}^{L} \mu_{\ell} + K^2\frac{M(1-l\frac{M}{N})}{KM+N}(\mu_1+\mu_{L}),
$
which can also be viewed as the objective function with parameters $\alpha_L = 1$ and $\gamma_L = \frac{M}{N}$. 

As shown in Fig.~\ref{fig:performance_comparison}, the baseline scheme exhibits a non-monotonic cost profile: the communication cost increases with $L$, peaking around $L=8$, and then decreases as $L$ grows further. This indicates that the baseline scheme suffers from substantial overhead when the access level is small to moderate. In contrast, the proposed superposition MACC scheme achieves significantly lower and more stable costs when $L < 16$. In particular, its communication cost remains nearly flat up to $L=13$, thereby eliminating the initial degradation observed in the baseline scheme. For $L > 16$, the cost of the proposed MACC scheme converges to that of the baseline scheme.

\section{Conclusion}
\label{sec: conc}

%In this paper, we studied the MACC model with heterogeneous retrieval costs, where both broadcast and cache access costs are taken into consideration. To address this problem, we proposed a novel coded caching framework based on superposition coding, where the MACC schemes of Cheng et al.~\cite{MAC} are layered across different access levels. The proposed scheme performs well in cost-aware MACC systems, and we prove that the optimal solution always exhibits a sparse structure with at most two non-zero components. Ongoing work includes extensions to the two-dimensional MACC model and the combinatorial MACC model with heterogeneous retrieval costs.

In this paper, we studied the MACC model with heterogeneous retrieval costs from an optimization perspective, accounting for both broadcast and cache access costs. To address the tradeoff between caching gains and retrieval costs, we proposed a novel coded caching framework based on superposition coding, where the MACC schemes of Cheng et al.~\cite{MAC} are layered across multiple access levels. The proposed scheme was shown to be more cost-efficient than the scheme in~\cite{MAC} in cost-aware MACC systems, and we proved that the optimal solution always exhibits a sparse structure with at most two non-zero components. Ongoing work includes the extensions to the two-dimensional MACC model and the combinatorial MACC model with heterogeneous retrieval costs.

\appendices

\section{Example for Construction \ref{constr-3-arrays}}
\label{sec: Cons1}

Now let us use the following example to illustrate Construction \ref{constr-3-arrays}.
\begin{example}
\label{ex: MACC}    
Consider the parameters $K'=4$, $t=2$, and $L=3$. From \eqref{eq-cach-node-C}, \eqref{eq-cach-user-U} and \eqref{eq-delivery-user-Q}, we are able to construct the node-placement array $\mathbf{C}$, user-retrieve array $\mathbf{U}$ and user-delivery array $\mathbf{Q}$ for a $(K=8,L=3,M=2,N=8)$ MACC system. When $(\mathcal{T},g)=(\{1,2\},1)$, from \eqref{eq-row-cache} and \eqref{eq-row-user} we have
 \begin{align*} 
\mathcal{C}_{\{1,2\},1}&= \{ \langle \mathcal{T}[h] + 2h \rangle_{8} : h \in [t]\} =\{3,6\},\\
 \mathcal{U}_{\{1,2\},1}&= \left\{ \langle \mathcal{T}[h] + 2(h-1)  \rangle_{8} ,  \right. \langle \mathcal{T}[h] + 2(h-1) + 1\rangle_{8},  \nonumber\\
&\ \ \ \  \left.  \ldots, \langle \mathcal{T}[h] + 2h \rangle_{8}: h \in [t]\right\} \\
&=\{1,2,3,4,5,6\}. 
\end{align*}
Now let us consider the mapping $\psi_{\mathcal{T},g}(\cdot)$. We have $[8]\setminus\mathcal{U}_{\{1,2\},1} = \{7,8\}$ and $[4] \setminus \{1,2\} = \{3,4\}$. For $k=7$ which is the first entry of $\{7,8\}$, we obtain $\psi_{\{1,2\},1}(7) = \{3,4\}[1] = \{3,4\}[1] = 3$. From \eqref{eq-delivery-user-Q}  we have $ \{1,2\} \cup \psi_{\{1,2\},1}(7) = \{1,2,3\}$ which means the entry in the location $\left((\{1,2\},1),7\right)$ is $(\{1,2,3\},1)$. Similarly, we have all the vectors concurring in $\mathbf{Q}$ are 
\begin{align}
\label{eq-vectors}
(\{1,2,3\},g), (\{1,2,4\},g),  (\{1,3,4\},g), (\{2,3,4\},g)
\end{align}where $g\in[8]$. To simplify the vectors in \eqref{eq-vectors}, we use the integers in $[32]$. For instance, $(\{1,2,3\},1)$, $(\{1,2,4\},1)$, $(\{1,3,4\},1)$, and $(\{2,3,4\},1)$ can be represented $1,2,3,4$.

\begin{figure*}[!t]
\centering
\centerline{\includegraphics[scale=0.48]{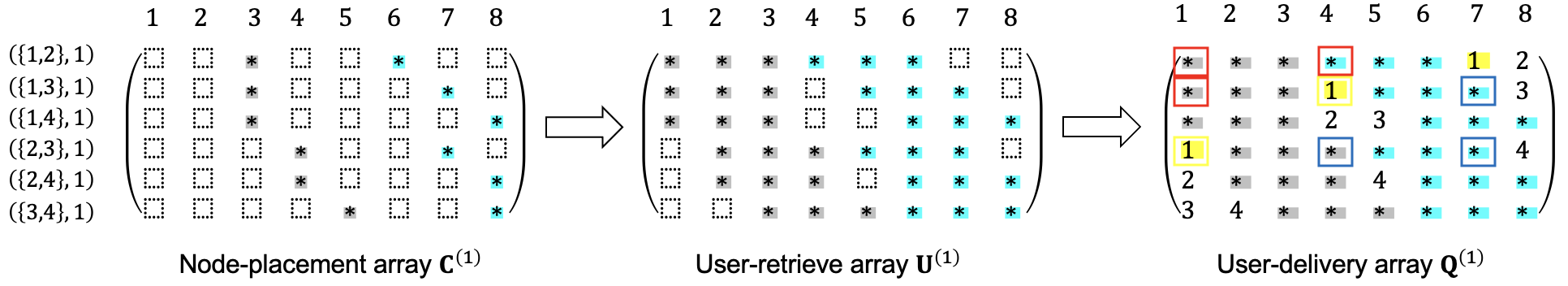}}
\caption{An example of $\mathbf{C}^{(1)}, \mathbf{U}^{(1)}, \mathbf{Q}^{(1)}$ under system parameters $(8,3,2,8)$.}
\label{fig: MACC construction}  
\end{figure*} 
For the limitation of pages, here we just list the subarrays $\mathbf{C}^{(1)}$, $\mathbf{U}^{(1)}$, and $\mathbf{Q}^{(1)}$ generated by the row labels $(\mathcal{T},1)$ where $\mathcal{T}\in {[4]\choose 2}$ and the column labels $[8]$ in Fig.~\ref{fig: MACC construction}. We can see that $\mathbf{Q}^{(1)}$ satisfies the conditions C$1$ and C$2$ in Lemma \ref{lem-sufficient-Q}. For instance, when $s=1$, we have the following array
\begin{align*}
 \left(\begin{array}{ccc}
  *&*&1 \\
  *&1&*\\
  1&*&*
 \end{array}  
 \right).
\end{align*}
The following investigations can be obtained. 
\begin{itemize}
    \item The subarrays $\mathbf{C}^{(g)}$, $\mathbf{U}^{(g)}$, and $\mathbf{Q}^{(g)}$ generated by the row labels $(\mathcal{T},g)$ where $\mathcal{T}\in {[4]\choose 2}$, $g\in[8]$ and the column labels $[K]$ can be obtained by cyclically right-shifting $\mathbf{C}^{(1)}$, $\mathbf{U}^{(1)}$, and $\mathbf{Q}^{(1)}$ by $g-1$ positions, respectively. That is, 
\begin{align*}
\mathbf{C}&=\left[
              \mathbf{C}^{(1)};
              \mathbf{C}^{(2)};
              \cdots;
              \mathbf{C}^{(8)}
            \right],\\
\mathbf{U}&=\left[
              \mathbf{U}^{(1)};
              \mathbf{U}^{(2)};
              \cdots;
              \mathbf{U}^{(8)}
            \right],\\
\mathbf{Q}&=\left[
              \mathbf{Q}^{(1)};
              \mathbf{Q}^{(2)};
              \cdots;
              \mathbf{Q}^{(8)}
            \right].   
\end{align*} Here, $\mathbf{Q}$ satisfies the conditions C$1$ and C$2$ in Lemma \ref{lem-sufficient-Q}.
\item Each user can only use the $1^{\text{th}}$ closed cache node and the $L^{\text{th}}$ closed cache node to access some packets to decode its required packets from the coded packets sent by the server. For instance, at the time slot $s=1$, the server send the coded packet $$X_1= W_{d_1,(\{1,2\},1)}\oplus W_{d_4,(\{1,3\},1)}\oplus W_{d_7,(\{2,3\},1)}.$$ Then in order to decode its required packet $W_{d_1,(\{1,2\},1)}$, the user $U_1$ should get the packets $W_{d_4,(\{1,3\},1)}$ (red block) and $W_{d_7,(\{2,3\},1)}$ (red block) from the $3^{\text{th}}$ closed cache node $C_3$ respectively; similarly to decode its required packet $W_{d_4,(\{1,3\},1)}$, the user $U_4$ should get the packets $W_{d_1,(\{1,2\},1)}$ (blue block) and $ W_{d_7,(\{2,3\},1)}$ (red block) from the $1^{\text{th}}$ closed cache node $C_4$ and the $3^{\text{th}}$ closed cache node $C_5$ respectively; 
to decode its required packet $W_{d_7,(\{2,3\},1)}$, the user $U_7$ should get the packets $W_{d_1,(\{1,2\},1)}$ (blue block) and $W_{d_4,(\{1,3\},1)}$ (blue block) from the $1^{\text{th}}$ closed cache node $C_4$ respectively. In fact, for the general case, this investigation always holds. For the detailed proof, please see Subsection \ref{sec: costs}.
\end{itemize}
\hfill $\square$
\end{example}

\section{MACC Scheme with Costs}
\label{sec: optimization problem establish}

This section examines the $(K,\boldsymbol{\mu},\rho, M, N)$ MACC model and presents the formulation of communication cost. Furthermore, an optimization problem is formulated based on superposition coding. 

\subsection{Communication Cost Formulation}
\label{sec: costs}

We formulate the communication cost of the $(K,\boldsymbol{\mu},\rho,M,N)$ MACC system by applying Construction~\ref{constr-3-arrays}, under system parameters $(K',t,L)$. The relationships of these parameters are $K = K' + t(L-1)$, $t = \frac{KM}{N}$, $\gamma = \frac{M}{N}$ and $L = |\boldsymbol{\mu}|$.  

To determine the cost formulation, the total cost is divided into two kinds of costs based on the file retrieval procedure. {\bf server broadcast cost}, the cost of the server broadcasting the broadcast messages. {\bf Cache-access cost}, consisting of two parts: the cost of users retrieving the requested files directly from cache nodes, and the cost of users retrieving the decoding files for the broadcast messages. 

$\bullet$ {\bf Server broadcast cost:} We consider the broadcast cost based on the user-delivery array $\mathbf{Q}$~\eqref{eq-delivery-user-Q}. Recall that for any $(\mathcal{T},g),$ $\mathcal{T} \in \binom{[K']}{t}, g,k \in [K]$, we have   
\begin{align}
    \mathbf{Q}((\mathcal{T},g),k)&= 
    \begin{cases}
        * & \text{if}\  k \in \mathcal{U}_{\mathcal{T},g} \\
        (\mathcal{T} \cup \{\psi_{\mathcal{T},g}(k)\},g) & \text{otherwise},
    \end{cases}  
\end{align}
where \begin{align*} \mathcal{U}_{\mathcal{T},g}&= \left\{ \langle \mathcal{T}[h] + (h-1)(L-1) + (g-1)\rangle_{K},  \right. \nonumber\\  & \ \ \ \ \ \langle \mathcal{T}[h] + (h-1)(L-1) + (g-1) + 1\rangle_{K},  \nonumber\\
&\ \ \ \  \left.  \ldots, \langle\mathcal{T}[h] + h(L-1) +(g-1)\rangle_{K}: h \in [t]\right\}.
\end{align*}
Considering the first round transmission, the server broadcasts $X_{(\mathcal{T}',g)}$ which is a XOR of all $(\mathcal{T}\cup\{\psi_{\mathcal{T},g}(k)\},g) = (\mathcal{T}',g)$ entries to all users
\begin{equation}
X_{(\mathcal{T}',g)} = \bigoplus_{\mathbf{Q}((\mathcal{T},1),k)=(\mathcal{T}',g)} W_{d_{U_k},(\mathcal{T},g)},
\end{equation}
where $\mathcal{T}' \in \binom{[K']}{t+1}$. The broadcast cost depends on the number of broadcast messages and the size of each broadcast message, with a per-file cost, $\rho$. We obtain that the number of transmission messages is $S = K \times \binom{K'}{t+1}$, and the size ratio of the broadcast message is equal to the file packet is $\frac{1}{K\binom{K'}{t}}$. Thus, the broadcast cost can be written as
\begin{align} \label{eq: R_b ge}
     R_b = S \frac{1}{K\binom{K'}{t}} \rho =  \frac{K\binom{K'}{t+1}}{K\binom{K'}{t}} \rho = K\frac{1-L\gamma}{1+K\gamma} \rho. 
\end{align}

$\bullet$ {\bf Cache-access cost (i):} 
In this part, we analyze the cost of users retrieving the requested packets directly from cache nodes. 
Each user $U_k$ retrieves the packets $W_{d_{U_k},(\mathcal{T},g)}$ corresponding to the entries where $\mathbf{Q}((\mathcal{T},g),k)=*$. 
The access cost depends on the relative position of $k$ inside the set $\mathcal{U}_{\mathcal{T},g}$.

Specifically, $\mathcal{U}_{\mathcal{T},g}$ can be decomposed into $t$ disjoint intervals:
\begin{align}
\mathcal{U}_{\mathcal{T},g}
&=\bigcup_{h=1}^{t} I_{\mathcal{T},g,h}, \nonumber  \\ 
I_{\mathcal{T},g,h} &\triangleq 
\big\{\langle \mathcal{T}[h]+ (h-1)(L-1)+ (g-1)+ r - 1\rangle_{K}: \nonumber  \\
& \ \ \ \ \ r\in[L]\big\},
\label{eq: interval-def}
\end{align}
where each interval $I_{\mathcal{T},g,h}$ corresponds to one specific $h\in[t]$ and contains $L$ consecutive user indices. If user $k$ belongs to the interval $I_{\mathcal{T},g,h}$, i.e., 
$$
k = \langle \mathcal{T}[h]+(h-1)(L-1)+(g-1)+r-1\rangle_{K},
$$
for some $r\in[L]$. Then, the cache node is the $(L-r+1)^{\text{th}}$ closed cache node for user $U_k$ and the corresponding access cost for this packet equals $\mu_{L-r+1}$.

Now, let's consider the total cost incurred by user $U_k$. As $g$ cycles over $[K]$, the $t$ intervals in $\mathcal{U}_{\mathcal{T},g}$ shift cyclically so that, for each $h\in[t]$, user $k$ occupies every position $r \in\{1,\ldots,L\}$ exactly once. Therefore, for every fixed $\mathcal{T}$, user $k$ incurs the position-dependent costs $\{\mu_{L-r+1}\}_{r=1}^{L}$ each exactly once per interval, yielding a per-$\mathcal{T}$ total $\sum_{r=1}^{L}\mu_{L-r+1} = \sum_{\ell=1}^{L}\mu_{\ell}$. Since there are $t$ intervals, the per-$\mathcal{T}$ cost becomes $t\sum_{\ell=1}^{L}\mu_{\ell}$. 
Summing over all $\binom{K'}{t}$ choices of $\mathcal{T}$, the total (unnormalized) cost for user $k$ equals
\begin{align}
R_{c_1,k} = \binom{K'}{t}\,t \sum_{\ell=1}^{L}\mu_{\ell},
\label{eq: per-user-u2c-raw}
\end{align}
which is identical for all users $k\in[K]$. W%ith the size ratio of each file packet, $\frac{1}{K\binom{K'}{t}}$, 
We then obtain the normalized access cost as
\begin{align} \label{eq: R_c1 ge}
    R_{c_1} = K \times R_{c_1,k} \times \frac{1}{K\binom{K'}{t}} 
    = t \sum_{\ell=1}^{L}\mu_{\ell} 
    = K \sum_{\ell=1}^{L}\mu_{\ell}\gamma.
\end{align}

$\bullet$ {\bf Cache-access cost (ii):} Finally, we analyze the cache-access cost incurred when users retrieve file packets required for decoding the broadcast messages. Consider a broadcast message associated with entry $(\mathcal{T}',g)$, where $\mathcal{T}'\in\binom{[K']}{t+1}$. Firstly, let us find out the retrieved file packets. In the user-delivery array $\mathbf{Q}$, there exist $t+1$ entries with the same value $(\mathcal{T}',g)$, located at positions ${((\mathcal{T}_u,g),k_u)}_{u=1}^{t+1}$, where $\mathcal{T}_u\in\binom{\mathcal{T}'}{t}$, $k_u\in[K]$, and $\mathcal{T}_u\cup{\psi_{\mathcal{T}_u,g}(k_u)}=\mathcal{T}'$. By condition C1, all pairs $(\mathcal{T}_u,g)$ and $k_u$ are distinct, so the subarray $\mathbf{Q}_{(\mathcal{T}',g)}$ formed by rows $\{(\mathcal{T}_1,g),(\mathcal{T}_2,g),\dots,(\mathcal{T}_{t+1},g)\}$ and columns $\{k_1,\dots,k_{t+1}\}$ is a $t+1 \times t+1$ matrix. By repeatedly applying condition C2 to any $u \neq v$, we have $\mathbf{Q}((\mathcal{T}_u,g),k_v)=*$. Hence, up to simultaneous row/column permutations, $\mathbf{Q}_{(\mathcal{T}',g)}$ takes the canonical form

\iffalse
\begin{align}
\mathbf{Q}_{(\mathcal{T}',g)} =
\begin{blockarray}{cccc}
  k_1 & k_2 & \cdots & k_{t+1} \\
  \begin{block}{(cccc)}
    (\mathcal{T}',g) & * & \cdots & * \\
    * & (\mathcal{T}',g) & \cdots & * \\
    \vdots & \vdots & \ddots & \vdots \\
    * & * & \cdots & (\mathcal{T}',g) \\
  \end{block}
\end{blockarray}.
\end{align}
\fi

Accordingly, user $U_{k_v}$, where $v \in [t+1]$, receives the broadcast message 
\begin{align*}
X_{(\mathcal{T}',g)} = \bigoplus_{\mathbf{Q}((\mathcal{T}_u,g),k_u)=(\mathcal{T}',g)} W_{d_{U_{k_u}},(\mathcal{T}_u,g)},
\end{align*}
and must retrieve the file packets $W_{d_{U_{k_u}},{(\mathcal{T}_u,g)}}: u \neq v$ from the corresponding cache nodes to decode its desired packet $W_{d_{U_{k_v}},(\mathcal{T}_v,g)}$.

Next, we will determine the cache-access cost for retrieving these packets. Recall that $\psi_{\mathcal{T},g}(\cdot)$ is a one-to-one mapping function which maps the user index $k$ into the $\langle h + (g-1) \rangle_{K'}$-th element of $[K'] \setminus \mathcal{T}$, where $k$ is the $h^{\text{th}}$ entry of $[K] \setminus \mathcal{U}_{\mathcal{T},g}$. Based on the definition of $\psi_{\mathcal{T},g}(\cdot)$ and definition of $\mathcal{U}_{\mathcal{T},g}$ in~\eqref{eq-row-user}, we can define the inverse mapping function $\psi^{-1}_{\mathcal{T},g}(\cdot)$ which maps the $\langle h + g-1 \rangle_{K'}$-th entry of $[K'] \setminus \mathcal{T}$ into $[K] \setminus \mathcal{U}_{\mathcal{T},g}$. Suppose $r$ is the $n^{\text{th}}$ smallest entry of $\mathcal{T} \cup \{r\}$, there are $n-1$ elements in $\mathcal{T}$ smaller than $r$, then we have 
\begin{align}
    \psi^{-1}_{\mathcal{T},g}(r) = \langle (n - 1)(L-1) +  r + g - 1  \rangle_{K}.
\end{align}
Let $\psi_{\mathcal{T}_v,g}(k_v) = r_v$ denote the $n_v^{\text{th}}$ smallest element of $\mathcal{T}'$, for all $v \in [t+1]$. We have 
\begin{align}
        & k_v = \psi^{-1}_{\mathcal{T}_v,g}(r_v) = \langle (n_v-1)(L-1) +  r_v + g - 1  \rangle_{K}, \\
     \text{and}~&\mathcal{T}_1 \cup r_1 = \mathcal{T}_2 \cup r_2 = \ldots  = \mathcal{T}_{t+1} \cup r_{t+1} = \mathcal{T}'.
\end{align}
Then we find out that $r_v \neq r_u$ for any $v \neq u$ and $\{r_1,r_2,\ldots,r_{t+1}\} = \mathcal{T}'$. Now, we will analyze the cost of user $U_{k_v}$. 
\begin{itemize}
    \item[\text{(a)}] \text{Case $r_u>r_v$:} Here $r_v$ is the $n_v^{\text{th}}$ smallest element in  $\mathcal{T}'$ and also the $n_v^{\text{th}}$ smallest element in $\mathcal{T}_u$. We want to find the cache node that caches the file packet $W_{d_{U_{k_v}}, (\mathcal{T}_u,g)}$. Based on the definition of $\mathbf{C}_{\mathcal{T},g}$ in~\eqref{eq-row-cache} and $\mathcal{U}_{\mathcal{T},g}$ in~\eqref{eq-row-user}, we find that 
    \begin{align}
        k_v &= \langle (n_v-1)(L-1) +  r_v + g - 1  \rangle_{K} \\
        &= \langle r_v + n_v(L-1) +   g - 1 - (L-1)  \rangle_{K}.
    \end{align}
    There are $L-1$ distance from the $n_v^{\text{th}}$ element of $\mathbf{C}_{\mathcal{T}_u,g}$, which means file packet $W_{d_{U_{k_v}},(\mathcal{T}_u,g)}$ is cached in the cache node $C_{\langle k_v + (L-1) \rangle_{K}}$, and the user $U_{k_v}$ can retrieve it with cost $\mu_L$.
 
     \item[\text{(b)}] \text{Case $r_u<r_v$:} Here $r_v$ is the $n_v^{\text{th}}$ smallest element in  $\mathcal{T}'$ and the $(n_v -1)^{\text{th}}$ smallest element in $\mathcal{T}_u$. We want to find the cache node that caches the file packet $W_{d_{U_{k_v}}, (\mathcal{T}_u,g)}$. With the same process, we find that 
    \begin{align}
        k_v &= \langle  r_v + (n_v-1)(L-1) +  g - 1  \rangle_{K}, 
    \end{align}
    which is exactly the $(n_v-1)^{\text{th}}$ element of $\mathbf{C}_{\mathcal{T}_u,g}$, which means file packet $W_{d_{U_{k_v}},(\mathcal{T}_u,g)}$ is cached in the cache node $C_{k_v}$, and the user $U_{k_v}$ can retrieve it with cost $\mu_1$.
\end{itemize}

Then we have the following conclusion: For any $\mathcal{T}' \in \binom{[K']}{t+1}$, user $U_{k_v}$, $v\in[t+1]$ where $\psi_{\mathcal{T}_v, g}(k_v) = r_v$, can decode the file packet $W_{d_{U_{k_v}},(\mathcal{T}_v, g)}$ from the broadcast coding signal 
$$X_{(\mathcal{T}',g)} = \bigoplus_{\mathbf{Q}((\mathcal{T}_u,g),{k_u})=(\mathcal{T}',g)} W_{d_{U_{k_u}},(\mathcal{T}_u, g)},$$
by retrieving two XOR of side file packets:
$$ \bigoplus_{r_u > r_v} W_{d_{U_{k_u}},(\mathcal{T}_u, g)} ~\text{and} \bigoplus_{r_u < r_v} W_{d_{U_{k_u}},(\mathcal{T}_u, g)}$$
from cache nodes $C_{\langle k_v+L-1 \rangle_K}$ and $C_{k_v}$, respectively, with cache-access costs $\mu_L$ and $\mu_1$. 

Specifically, suppose $r_1>r_2>\ldots>r_{t+1}$, we will obtain that user $U_{k_1}$ retrieves $\bigoplus_{u \in [2:t+1]} W_{d_{U_{k_u}},(\mathcal{T}_u, g)}$ at cost $\frac{\mu_1}{K\binom{K'}{t}}$. User $U_{k_v}$, where $v \in [2:t]$, retrieves $\bigoplus_{u \in [v-1]} W_{d_{U_{k_u}},(\mathcal{T}_u, g)}$ and $\bigoplus_{u \in [v+1:t+1]} W_{d_{U_{k_u}},(\mathcal{T}_u, g)}$ at cost $\frac{\mu_1 + \mu_L}{K\binom{K'}{t}}$. User $U_{k_{t+1}}$ retrieves $\bigoplus_{u \in [t]} W_{d_{U_{k_u}},(\mathcal{T}_u, g)}$ at cost $\frac{\mu_L}{K\binom{K'}{t}}$.

The cost for users to access cache nodes in order to decode the requested file packets from $S = K\binom{K'}{t+1}$ broadcast messages is given by
\begin{align} \label{eq: R_c2 ge}
R_{c_2} &= \frac{St(\mu_1+\mu_L)}{K\binom{K'}{t}} \nonumber \\ &= \frac{K'-1}{t+1}t(\mu_1+\mu_L) \nonumber \\ &= K^2\gamma\frac{1-L\gamma}{K\gamma+1}(\mu_1+\mu_L).
\end{align}

Based on equations~\eqref{eq: R_b ge},~\eqref{eq: R_c1 ge} and~\eqref{eq: R_c2 ge}, the total system cost of the MACC scheme obtained with system parameters $(K,L,M,N,\boldsymbol{\mu}=(\mu_1,\ldots,\mu_L),\rho)$ can be calculated as
\begin{align} \label{eq: R ge}
R &= R_b + R_{c_1} + R_{c_2} \nonumber \\ &= K\frac{1-L\gamma}{K\gamma+1}\rho +   K\gamma \sum_{l=1}^{L} \mu_l +  K^2\gamma\frac{1-L\gamma}{K\gamma+1}(\mu_1+\mu_{L}).
\end{align}

\bibliographystyle{IEEEtran}
\bibliography{references_d2d}

%\bibliographystyle{IEEEbib}
%\bibliography{references_d2d}

\end{document}

%% file: macros.tex
\long\def\comment#1{}

\newtheorem{example}{Example}%[section]
\newtheorem{theorem}{Theorem}%[section]
\newtheorem{lemma}{Lemma}%[section]
\newtheorem{proposition}{Proposition}%[section]
\newfont{\bbb}{msbm10 scaled 700}

\newfont{\bb}{msbm10 scaled 1100}

\newcommand{\be}{\begin{equation}}
\newcommand{\ee}{\end{equation}}
\newcommand{\bea}{\begin{eqnarray}}
\newcommand{\eea}{\end{eqnarray}}